\documentclass[letterpaper, 10 pt, conference]{ieeeconf}  

\IEEEoverridecommandlockouts                              

\usepackage{graphics} 
\usepackage{amsmath} 
\usepackage{amssymb}  
\usepackage{tikz}
\usepackage{pgfplots}
\pgfplotsset{width=\columnwidth,compat=1.18} 

\usepackage{tikz}
\usepackage{tikz-cd}
\usetikzlibrary{patterns,quotes,angles}

\newtheorem{theorem}{Theorem}
\newtheorem{proposition}{Proposition}
\newtheorem{definition}{Definition}
\newtheorem{remark}{Remark}

\title{\LARGE \bf Dissipation-assisted stabilization of periodic orbits via actuated exterior impacts in hybrid mechanical systems with symmetry
}

\author{William Clark$^{1}$, Leonardo Colombo$^{2}$, and Anthony Bloch$^{3}$
\thanks{*W. Clark was funded by AFOSR grant FA9550-23-1-0400. L. Colombo acknowledges financial support from Grant PID2022-137909NB-C21 funded by MCIN/AEI/ 10.13039/501100011033 and  iRoboCity2030-CM, Robótica Inteligente para Ciudades Sostenibles (TEC-2024/TEC-62). A.Bloch was partially supported by NSF grant  DMS-2103026, and AFOSR grants FA
9550-22-1-0215 and FA 9550-23-1-0400.}
\thanks{$^{1}$William Clark is with the Department of Mathematics,
        Ohio University, Athens, OH, 45701, USA
        {\tt\small clarkw3@ohio.edu}}%
\thanks{$^{2}$L. Colombo is with Centre for Automation and Robotics (CSIC-UPM), Ctra. M300 Campo Real, Km 0,200, Arganda del Rey - 28500 Madrid, Spain. {\tt\small leonardo.colombo@csic.es}}
\thanks{$^{3}$A. Bloch is with Department of Mathematics, University of Michigan, Ann Arbor, MI 48109, USA. {\tt\small abloch@umich.edu}}
}

\begin{document}

\maketitle
\thispagestyle{empty}
\pagestyle{empty}

\begin{abstract}
Impulsive mechanical systems exhibit discontinuous jumps in their state, and when such jumps are triggered by spatial events, the geometry of the impact surface carries information about the controllability of the hybrid dynamics. For mechanical systems defined on principal $G$-bundles, two qualitatively distinct types of impacts arise: interior impacts, associated with events on the shape space, and exterior impacts, associated with events on the fibers. A key distinction is that interior impacts preserve the mechanical connection, whereas exterior impacts generally do not.  In this paper, we exploit this distinction by allowing actuation through exterior impacts. We study the pendulum-on-a-cart system, derive controlled reset laws induced by moving-wall impacts, and analyze the resulting periodic motions. Our results show that reset action alone does not provide a convincing stabilizing regime, whereas the addition of dissipation in the continuous flow yields exponentially stable periodic behavior for suitable feedback gains.

\end{abstract}

\section{Introduction}

Hybrid dynamical systems combine continuous-time evolution with discrete state resets and arise naturally in a broad range of applications, including robotic locomotion, mechanical systems with impulses, and multi-agent control \cite{Johnson1994,Goebel2012,vanderSchaft2000,Westervelt2007,Holmes2006,Lee2013}.
In many of these problems, the continuous dynamics are not arbitrary but inherit geometric structure from an underlying mechanical system.
In particular, when the configuration space carries symmetries, the dynamics is naturally organized by objects such as momentum maps, reduced spaces, and mechanical connections.
Understanding how these structures interact with impacts is essential for the analysis and control of hybrid mechanical systems.

Over the last years, several works have shown that symmetry plays a nontrivial role in hybrid dynamics.
In hybrid mechanical systems, impacts do not merely interrupt the continuous flow: they may preserve, alter, or obstruct the geometric quantities associated with symmetry, depending on the structure of the switching surface and the reset map.
This perspective has motivated a broader geometric program on hybrid mechanical systems with symmetry, including hybrid reduction, hybrid momentum maps, and the study of periodic orbits in simple hybrid mechanical systems \cite{Ames2006,ColomboIrazu2020,colombo2025generalized,irazu2021reduction,BlochClarkColombo2017}.
From this viewpoint, one is led to ask not only whether symmetries survive the hybrid evolution, but also how they constrain the possibility of using impacts as a control mechanism.

Motivated by these observations, in this paper we build on our previous analysis of impacts in hybrid mechanical systems with symmetry \cite{10591262}, where we showed that the geometry of the switching surface determines whether the mechanical connection is preserved across impact. In particular, impacts supported on vertical surfaces preserve the connection-induced symmetry variable, whereas exterior impacts generally do not. This naturally leads to the geometric control question addressed here: \textit{can the symmetry-modifying effect of exterior impacts be exploited for the generation and stabilization of periodic hybrid motions}?

We study this question in hybrid mechanical systems with symmetry where reset actuation is applied through exterior impacts and dissipation is introduced in the continuous flow. The reset provides directional correction in the symmetry variable, while the dissipative flow induces contraction between impacts. We develop this idea in the pendulum-on-a-cart system, viewed as a mechanical system on a principal bundle. Building on its role in \cite{10591262} as a model for the distinction between interior and exterior impacts, we derive controlled reset laws associated with moving-wall impacts, analyze the resulting periodic orbits, and study their orbital stability through the linearized Poincar\'e map. These findings suggest that exterior reset actuation alone is not sufficient to robustly stabilize the orbit, whereas dissipation in the continuous flow creates the additional contraction needed for orbital stability.

The main contribution of this paper is to show that the geometric distinction between interior and exterior impacts has direct consequences for control in hybrid mechanical systems with symmetry. In particular, exterior impacts allow the reset to act on the symmetry direction, and, when combined with dissipation in the continuous flow, this yields a mechanism for the generation and stabilization of periodic hybrid motions. We develop this idea for the pendulum-on-a-cart system by deriving controlled moving-wall reset laws and analyzing the resulting closed-loop periodic orbits.

The remainder of the paper is organized as follows. Section~II presents the hybrid Hamiltonian and principal-bundle framework. Section~III studies interior and exterior impacts and derives the controlled reset laws. Section~IV analyzes periodic orbit generation and dissipation-assisted stabilization. Section~V concludes the paper.

\section{Hybrid Mechanical Systems with Symmetries}

\subsection{Hybrid Hamiltonian systems}

We begin with the hybrid setting, where the continuous evolution is Hamiltonian and the discrete transitions are described by state-triggered reset/impact maps. A simple hybrid system is specified by a tuple
$\mathcal H = (\mathcal X,\mathcal S,X,\Delta)$, where $\mathcal X$ is a smooth manifold, $X$ is a $C^1$ vector field on $\mathcal X$, $\mathcal S\subset \mathcal X$ is an embedded codimension-one submanifold called the \emph{switching or impact surface}, and $\Delta:\mathcal S\to \mathcal X$ is a smooth embedding called the \emph{reset or impact map}. The associated hybrid dynamical system is
\begin{equation}\label{eq:hybrid_general}
\Sigma:
\begin{cases}
\dot x = X(x), & x\notin \mathcal S,\\
x^+ = \Delta(x^-), & x^- \in \mathcal S.
\end{cases}
\end{equation}

In the Hamiltonian case, the state space is $\mathcal X=T^*Q$ and $X=X_H$ is the Hamiltonian vector field associated with a Hamiltonian $H:T^*Q\to\mathbb R$. Then \eqref{eq:hybrid_general} takes the form
\begin{equation}\label{eq:hybrid_hamiltonian}
\Sigma_H:
\begin{cases}
\dot z = X_H(z), & z\notin \mathcal S_H,\\
z^+ = \Delta_H(z^-), & z^- \in \mathcal S_H,
\end{cases}
\end{equation}
with $z=(q,p)\in T^*Q$ and where $\mathcal S_H\subset T^*Q$ is the switching surface and $\Delta_H:\mathcal S_H\to T^*Q$ is the impact map.

To exclude Zeno behavior, we assume that the set of impact times is closed and discrete, and $\overline{\Delta_H(\mathcal S_H)}\cap \mathcal S_H=\emptyset$.


\subsection{Impact geometry on principal bundles}

We now equip the hybrid Hamiltonian framework introduced above with symmetry. Let $Q$ be a configuration manifold carrying a free and proper left action of a Lie group $G$. Then the orbit space $Q/G$ is a smooth manifold, and the canonical projection $\pi:Q\to Q/G$ defines a principal $G$-bundle. The base manifold $Q/G$ is interpreted as the \emph{shape space}, whereas the fibers encode the symmetry directions. This viewpoint is standard in geometric mechanics and nonlinear control with symmetry; see, for instance, \cite{BL,grizzle2003structure}.

Throughout the paper, we consider systems associated with $G$-invariant mechanical structures on $Q$. The kinetic-energy metric induces the corresponding mechanical connection, and hence the decomposition of tangent vectors into horizontal and vertical components; see Sections 2.9 and 3.10 of \cite{BL}. In particular, the vertical bundle is
$\mathcal V_q := \ker T_q\pi \subset T_qQ$, for $q\in Q$. At the configuration level, impacts are described by an embedded codimension-one submanifold $\mathcal S\subset Q$, the impact surface. Its geometric position relative to the bundle projection $\pi$ will be fundamental in what follows.

\begin{definition}
An impact surface $\mathcal S\subset Q$ is said to be \emph{vertical} if there exists an embedded codimension-one submanifold $\Sigma\subset Q/G$ such that $\mathcal S=\pi^{-1}(\Sigma)$.
\end{definition}

\begin{definition}
An impact surface $\mathcal S\subset Q$ is said to be \emph{horizontal} if $(T\mathcal S)^\perp \subset \mathcal V$, where orthogonality is taken with respect to the kinetic energy metric and $\mathcal V=\ker T\pi$ is the vertical bundle.
\end{definition}

Thus, vertical impact surfaces are lifts of codimension-one submanifolds of the shape space and hence correspond to events determined purely by shape variables. By contrast, non-vertical impact surfaces are not induced from the shape space and may involve the symmetry variables explicitly. Horizontal surfaces form a distinguished subclass of non-vertical surfaces, characterized by the fact that their normal directions lie entirely in the vertical bundle.


\subsection{Controlled impacts and problem statement}

We allow the reset law to depend on a control input applied at impact. Thus, instead of a fixed reset map $\Delta_H$, we consider a family of controlled impact maps
\[
z^+ = \Delta_{H,u}(z^-), \qquad z^- \in \mathcal S_H,
\]
where $u$ denotes the control parameter. Our goal is to understand when such inputs can be used to generate and stabilize periodic orbits in hybrid mechanical systems with symmetry. In particular, we ask how the geometry of the switching surface constrains the ability of the impact map to modify the symmetry direction, and under what additional mechanisms it can contribute to stabilization.

We study this question on principal bundles, with particular emphasis on the distinction between vertical (interior) and non-vertical (exterior) impacts. As we shall see, this distinction determines whether the impact map can alter the connection-induced symmetry variable across impacts, and hence whether it can be used as a mechanism for the generation and stabilization of periodic motions.

\subsection{Example: pendulum on a cart}

We consider the pendulum-on-a-cart; see Fig.~\ref{fig:pendulum_cart}. Its configuration space is $Q=S^1\times \mathbb R$, with coordinates $(\theta,x)$, where $\theta$ denotes the pendulum angle and $x$ the cart position. The system is invariant under translations in the cart coordinate, and therefore admits the structure of a principal $\mathbb R$-bundle $\pi:Q=S^1\times \mathbb R \to S^1$, $(\theta,x)\mapsto \theta$, where the group action is given by $s\cdot (\theta,x)=(\theta,x+s)$ with $s\in \mathbb R$. Thus, the base space $Q/\mathbb{R}\simeq S^1$ is the shape space, while the fiber variable $x$ encodes the symmetry direction.

\begin{figure}[h!]
\centering
\begin{tikzpicture}[scale=1.0, every node/.style={font=\small}, line cap=round, line join=round]

    \draw[thick] (-3,-0.55) -- (3.5,-0.55);
    \foreach \x in {-2.9,-2.7,...,3.3}
        \draw[gray] (\x,-0.55) -- ++(-0.10,-0.16);

    \draw[thick, rounded corners=2pt, fill=gray!10] (-0.75,-0.12) rectangle (0.75,0.58);
    \node at (0,0.23) {$M$};

    \draw[thick, fill=white] (-0.42,-0.12) circle [radius=0.15];
    \draw[thick, fill=white] ( 0.42,-0.12) circle [radius=0.15];
    \fill (-0.42,-0.12) circle [radius=0.04];
    \fill ( 0.42,-0.12) circle [radius=0.04];

    \fill (0,0.58) circle [radius=0.05];

    \draw[dashed] (0,0.58) -- (0,-2.55);

    \draw[very thick] (0,0.58) -- (1.6,-1.95);
    \draw[thick, fill=gray!15] (1.6,-1.95) circle [radius=0.24];
    \node at (1.6,-1.95) {$m$};

    \node[right] at (0.92,-0.75) {$\ell$};

    \path
        (0,-1.3) coordinate (a)
        -- (0,0.58) coordinate (b)
        -- (1.6,-1.95) coordinate (c)
        pic[
            draw=black,
            ->,
            "$\theta$",
            angle eccentricity=1.3,
            angle radius=0.72cm
        ] {angle=a--b--c};

    \draw[->, thick] (-1.9,1.08) -- (-0.45,1.08);
    \node[above] at (-1.18,1.08) {$x$};

\end{tikzpicture}
\caption{Pendulum on a cart. The cart position $x$ defines the symmetry direction, while $\theta$ is the pendulum angle and $\ell$ the rod length.}
\label{fig:pendulum_cart}
\end{figure}
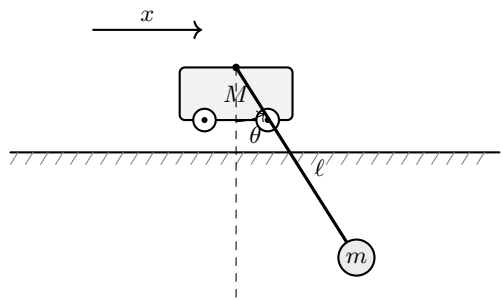

The (hyperregular) Lagrangian of the system is
\begin{align*}
L(\theta,x,\dot\theta,\dot x)
= &\frac12\Big(m\ell^2\dot\theta^2+2m\ell\cos\theta\,\dot x\dot\theta +(M+m)\dot x^2\Big)\\
&+mg\ell\cos\theta.
\end{align*}
The corresponding kinetic-energy metric on $Q$ is
\begin{align*}
g
=& m\ell^2\, d\theta\otimes d\theta
+m\ell\cos\theta\,(dx\otimes d\theta+d\theta\otimes dx)\\&
+(M+m)\, dx\otimes dx.
\end{align*}

The momenta conjugate to $(\theta,x)$ are
\begin{equation}\label{eq:cart_momenta}
\begin{aligned}
p_\theta &= m\ell^2\dot\theta + m\ell\cos\theta\,\dot x,\\
p_x &= m\ell\cos\theta\,\dot\theta + (M+m)\dot x.
\end{aligned}
\end{equation}
Using the inverse Legendre transform, one obtains
\begin{equation}\label{eq:cart_velocities}
\begin{aligned}
\dot\theta &= \frac{(M+m)p_\theta - m\ell\cos\theta\, p_x}{m\ell^2(M+m\sin^2\theta)},\\
\dot x &= \frac{-m\ell\cos\theta\, p_\theta + m\ell^2 p_x}{m\ell^2(M+m\sin^2\theta)}.
\end{aligned}
\end{equation}
Hence the Hamiltonian is
\begin{align}
H(\theta,x,p_\theta,p_x)
=&
\frac{(M+m)p_\theta^2 - 2m\ell\cos\theta\, p_\theta p_x + m\ell^2 p_x^2}
{2m\ell^2(M+m\sin^2\theta)}\nonumber\\
&-mg\ell\cos\theta.\label{eq:cart_hamiltonian}
\end{align}

Since the system has symmetry, the mechanical connection provides the natural way to separate motion along the base from motion along the fiber. The vertical bundle of the principal bundle $\pi:Q\to S^1$ is
\[
\mathcal V_{(\theta,x)}
=
\ker T_{(\theta,x)}\pi
=
\mathrm{span}\!\left\{\frac{\partial}{\partial x}\right\}
\subset T_{(\theta,x)}Q.
\]

The locked inertia tensor measures the kinetic energy along the symmetry directions and identifies the infinitesimal generator of the group action with the associated momentum. In the present case, since the symmetry group is $\mathbb R$, it reduces to the scalar $\mathbb I=M+m$. Hence, the associated mechanical connection is the $\mathbb R$-valued one-form $\mathcal A:TQ\to\mathbb R$ given by
\[
\mathcal A_{(\theta,x)}(\dot\theta,\dot x)
=
\frac{1}{M+m}\Big((M+m)\dot x+m\ell\cos\theta\,\dot\theta\Big).
\] Equivalently, in momentum coordinates, using \eqref{eq:cart_momenta},
$\displaystyle{\mathcal A=\frac{p_x}{M+m}}$. Accordingly, the horizontal space is
\[
\mathcal H_{(\theta,x)}
=
\ker \mathcal A_{(\theta,x)}
=
\mathrm{span}\!\left\{
\frac{\partial}{\partial\theta}
-\frac{m\ell\cos\theta}{M+m}\frac{\partial}{\partial x}
\right\}.
\]

This example is particularly useful because it naturally admits two geometrically distinct codimension-one impact surfaces. The first is
$\mathcal S_{\mathrm{int}}=\{\alpha\}\times\mathbb{R}$,
which is vertical, since it is the lift of the codimension-one submanifold $\{\alpha\}\subset S^1$:
$\mathcal S_{\mathrm{int}}=\pi^{-1}(\{\alpha\})$. Geometrically, this corresponds to an event determined purely by the shape variable $\theta$. By contrast, $\mathcal S_{\mathrm{ext}}=S^1\times\{z\}$, is not induced from the base space, since it fixes the fiber variable $x$ instead. It therefore represents an event acting directly on the symmetry direction. The corresponding switching surfaces in phase space are $\mathcal S_{H,\mathrm{int}}:=T^*Q|_{\mathcal S_{\mathrm{int}}}$,
and 
$\mathcal S_{H,\mathrm{ext}}:=T^*Q|_{\mathcal S_{\mathrm{ext}}}$.

\section{Interior and Exterior Impacts}

\subsection{Interior impacts and geometric obstruction}

We begin with impacts supported on vertical impact surfaces. In the present principal-bundle setting, an interior impact surface is precisely a vertical one, namely a surface of the form $\mathcal S=\pi^{-1}(\Sigma)$, for some codimension-one submanifold $\Sigma\subset Q/G$. Thus, interior impacts are determined entirely by the shape variables and do not act directly on the symmetry directions.

It was shown in \cite{10591262} that verticality of the impact surface is equivalent to preservation of the mechanical connection across impact. The same geometric argument applies to controlled interior resets, provided the reset remains compatible with the bundle structure and satisfies the same equivariance assumptions. Accordingly, if $\iota:\mathcal S\hookrightarrow Q$ denotes the canonical inclusion and $\Delta_u:TQ|_{\mathcal S}\to TQ$ is a controlled interior impact map, then the mechanical connection is preserved across impacts, in the sense that
\begin{equation}\label{eq:connection_preservation_interior}
\Delta_u^*\mathcal A=\iota^*\mathcal A.
\end{equation}
Equation~\eqref{eq:connection_preservation_interior} identifies the geometric obstruction associated with interior impacts: although the reset may alter the shape velocity, it does not provide a new mechanism for modifying the symmetry direction. In this sense, interior impacts do not provide impulsive control over the connection-induced symmetry variable.

For the pendulum-on-a-cart example, the interior impact surface is $\mathcal S_{\mathrm{int}}=\{\alpha\}\times\mathbb R$. The corresponding elastic impact map in momentum coordinates is
\begin{equation}\label{eq:inner_reset_momentum}
\Delta_H^{\mathrm{int}}:
\begin{cases}
p_\theta^+ = -p_\theta^- + \dfrac{2m\ell}{M+m}\cos\theta\, p_x^-,\\[1ex]
p_x^+ = p_x^-.
\end{cases}
\end{equation}
Equivalently, in velocity coordinates,
\begin{equation}\label{eq:inner_reset_velocity}
\Delta^{\mathrm{int}}:
\begin{cases}
\dot\theta^+ = -\dot\theta^-,\\[1ex]
\dot x^+ = \dot x^- + \dfrac{2m\ell}{M+m}\cos\theta\,\dot\theta^-.
\end{cases}
\end{equation}

This obstruction is transparent in this example. Since the reset satisfies $p_x^+=p_x^-$, the symmetry momentum is preserved across impact.  Moreover, because in momentum coordinates $\mathcal A=\frac{p_x}{M+m}$, it follows that $(\Delta_H^{\mathrm{int}})^*\mathcal A=\mathcal A|_{\mathcal S_{\mathrm{int}}}$. Therefore, in the interior-impact case, the reset does not alter the symmetry and hence does not provide a mechanism for controlling the symmetry direction through impacts.




\subsection{Exterior impacts and controlled resets}

We now turn to non-vertical impacts, which in the present setting will be referred to as \emph{exterior impacts}. In contrast with the interior case, exterior impacts act directly on the symmetry variable and therefore need not preserve the mechanical connection. This is precisely the mechanism that makes them relevant for impact-based control.

A particularly simple exterior impact surface for the pendulum-on-a-cart system is the fixed-wall guard
$\mathcal S_{\mathrm{ext}}=S^1\times\{\pm x^*\}$, corresponding to impacts with walls placed at the cart positions $x=\pm x^*$. For this switching surface, the elastic impact map in momentum coordinates is given by
\begin{equation}\label{eq:elastic_exterior_reset}
\Delta_H^{\mathrm{ext}}:
\begin{cases}
p_\theta^+ = p_\theta^-,\\
p_x^+ = -p_x^- + \dfrac{2}{\ell}\cos\theta\,p_\theta^-.
\end{cases}
\end{equation}
Equivalently, in velocity coordinates,
\begin{equation}\label{eq:elastic_exterior_reset_velocity}
\Delta^{\mathrm{ext}}:
\begin{cases}
\dot\theta^+ = \dot\theta^- + \dfrac{2}{\ell}\cos\theta\,\dot x^-,\\
\dot x^+ = -\dot x^-.
\end{cases}
\end{equation}
This reset does not preserve the connection-induced symmetry variable. Indeed, since $\mathcal A=\frac{p_x}{M+m}$, one has $(\Delta_H^{\mathrm{ext}})^*\mathcal A
=
\frac{-p_x+\frac{2}{\ell}\cos\theta\,p_\theta}{M+m}$,
which is not equal to $\mathcal A|_{\mathcal S_{\mathrm{ext}}}
=
\frac{p_x}{M+m}$.

To introduce actuation through the impact, we allow the wall to move with velocity $v$ at the instant of collision. This is realized by the Weierstrass-Erdmann corner conditions (see \S 4.4 of \cite{kirk2004optimal} or \S 3.5 of \cite{brogliato1999nonsmooth}):
\begin{equation}\label{eq:corner_conditions_moving_wall}
\begin{aligned}
p_x^+ &= p_x^- + \varepsilon,\\
H^- - H^+ &= \varepsilon\, v,
\end{aligned}
\end{equation}
where $\varepsilon$ denotes the impulsive exchange in the symmetry momentum. Solving these conditions yields the controlled impact map
\begin{equation}\label{eq:controlled_exterior_reset}
\Delta_{H,v}^{\mathrm{ext}}:
\begin{cases}
p_x^+ = -p_x^- + \dfrac{2}{\ell}\cos\theta\,p_\theta^- - 2v(M+m\sin^2\theta).\\
p_\theta^+ = p_\theta^-,
\end{cases}
\end{equation}
In velocity coordinates, this becomes
\begin{equation}\label{eq:controlled_exterior_reset_velocity}
\Delta_v^{\mathrm{ext}}:
\begin{cases}
\dot\theta^+ = \dot\theta^- + \dfrac{2}{\ell}(v+\dot x^-)\cos\theta,\\
\dot x^+ = -\dot x^- - 2v.
\end{cases}
\end{equation}

Equation~\eqref{eq:controlled_exterior_reset} shows that exterior impacts provide direct access to the symmetry momentum through the impact parameter $v$. This suggests a more general point of view: rather than treating $v$ merely as a physical parameter, one may regard it as a control input for assigning the post-impact symmetry variable. This naturally leads to the following question: \textit{can one prescribe the post-impact value of the connection-induced symmetry variable by a suitable choice of the impact control}? Here the desired value $\Gamma$ represents the symmetry value to be enforced immediately after impact, encoding a design objective such as orbit closure or feedback correction near a target periodic motion.

\begin{proposition}\label{prop1}
Let $\pi:Q\to Q/G$ be a principal $G$-bundle with mechanical connection $\mathcal A$, and let $\mathcal S_{\mathrm{ext}}\subset Q$ be a non-vertical impact surface. Assume that the controlled exterior impact map $\Delta_u:TQ|_{\mathcal S_{\mathrm{ext}}}\to TQ$ depends smoothly on a scalar control parameter $u$, and that the induced post-impact symmetry variable
$\Sigma_u:=\Delta_u^*\mathcal A$ depends affinely on $u$, namely $\Sigma_u(q,v)=F(q,v)+B(q,v)\,u$, with $B(q,v)\neq 0$ for all $(q,v)\in TQ|_{\mathcal S_{\mathrm{ext}}}$ in the region of interest. Assume moreover that $\dim\mathfrak g=1$. 

Then, for any prescribed smooth post-impact value $\Gamma:TQ|_{\mathcal S_{\mathrm{ext}}}\to\mathfrak g$, there exists a unique feedback law $u^\star(q,v)$ such that
$\Sigma_{u^\star}(q,v)=\Gamma(q,v)$ for all $(q,v)\in TQ|_{\mathcal S_{\mathrm{ext}}}$.
\end{proposition}

\begin{proof}
Since $\dim\mathfrak g=1$, the Lie algebra $\mathfrak g$ is isomorphic to $\mathbb R$. Thus, for each $(q,v)\in TQ|_{\mathcal S_{\mathrm{ext}}}$, the relation
\[
\Sigma_u(q,v)=F(q,v)+B(q,v)\,u
\]
is an affine scalar equation in the unknown $u$. Given a prescribed target $\Gamma(q,v)\in\mathfrak g\simeq\mathbb R$, the condition $\Sigma_u(q,v)=\Gamma(q,v)$ is equivalent to $F(q,v)+B(q,v)\,u=\Gamma(q,v)$.

Since $B(q,v)\neq 0$ this equation admits the unique solution
\[
u^\star(q,v)=\frac{\Gamma(q,v)-F(q,v)}{B(q,v)}.
\]
Because $F$, $B$, and $\Gamma$ are smooth and $B$ does not vanish, the resulting map $u^\star:TQ|_{\mathcal S_{\mathrm{ext}}}\to\mathbb R$ is smooth. By construction,
$\Sigma_{u^\star}(q,v)=\Gamma(q,v)$ for all $(q,v)\in TQ|_{\mathcal S_{\mathrm{ext}}}$.
\end{proof}

\begin{remark}
The assumption $\dim\mathfrak g=1$ is natural in many mechanical systems with a single symmetry direction, such as translation along one axis or rotation about a single angle. In this case, the connection-induced symmetry variable is scalar, and a scalar impact input is enough to assign its post-impact value pointwise. For higher-dimensional symmetry groups, the same construction would require either a vector-valued impact control or a suitable rank condition on the control action.\hfill$\diamond$
\end{remark}

For the pendulum-on-a-cart system with moving-wall impacts, the controlled exterior reset satisfies \eqref{eq:controlled_exterior_reset}. Since $\mathcal A=\frac{p_x}{M+m}$, it follows that
\[
(\Delta_{H,v}^{\mathrm{ext}})^*\mathcal A
=
\frac{-p_x+\frac{2}{\ell}\cos\theta\,p_\theta - 2v(M+m\sin^2\theta)}{M+m}.
\]
Thus, the post-impact connection-induced symmetry variable depends affinely on $v$ with nonvanishing coefficient, and therefore it can be assigned pointwise by a unique choice of the wall velocity.

Proposition \ref{prop1} shows that under a nondegeneracy condition, controlled exterior impacts provide a pointwise mechanism for assigning the post-impact symmetry variable. This does not by itself guarantee the existence of periodic hybrid motions, since periodicity also depends on the continuous-time evolution between impacts. However, it identifies the geometric degree of freedom through which such motions can be generated. In the pendulum-on-a-cart example, this mechanism is realized explicitly through the moving-wall velocity, which will be used in the next section to construct and analyze periodic hybrid orbits.

\section{Periodic Orbit Generation and Stability}

\subsection{Floquet analysis of hybrid orbits}

To study the stability of a periodic orbit in a hybrid system, one must combine the linearized continuous dynamics with the linearization of the reset at each impact. Consider a hybrid dynamical system of the form
\[
\mathcal{H} :
\begin{cases}
\dot{x} = f(x), & h(x)\neq 0, \\
x^+ = \Delta(x^-), & h(x^-)=0,
\end{cases}
\]
where the switching surface is locally described by the zero level set of a smooth function $h$.

Let $\gamma$ be a periodic trajectory with impact times $t_j$, so that $h(\gamma(t_j))=0$ for each $j$. Along the continuous portions of the orbit, the linearized dynamics is described by
\[
\dot{\Phi}(t)=A(t)\Phi(t), \quad A(t)=Df(\gamma(t)).
\]
At each impact time $t_j$, the state undergoes a jump of the form $$\Phi(t_j^+)=\Xi_j\,\Phi(t_j^-),$$
where $\Xi_j$ is the Jacobian associated with the reset map and the impact geometry.

When the switching surface is locally given by $h(x)=0$, the corresponding hybrid Jacobian is
\[
\Xi
=
D\Delta(x^-)\left[I-\frac{f(x^-)\otimes dh(x^-)}{dh(x^-)f(x^-)}\right]
+
\frac{f(x^+)\otimes dh(x^-)}{dh(x^-)f(x^-)}.
\]
Thus, the linearization of the hybrid flow over one period is obtained by concatenating the continuous fundamental solutions with the hybrid Jacobians at the successive impacts.

The resulting monodromy matrix defines the linearization of the Poincar\'e return map. Its eigenvalues are the Floquet multipliers of the hybrid orbit and determine its local orbital stability: if all nontrivial multipliers lie strictly inside the unit disk, then the periodic orbit is exponentially stable \cite{Westervelt2007}, \cite{goodman2019existence}.






\subsection{Dissipation-assisted stabilization}

The previous section shows that controlled exterior impacts provide a mechanism for assigning the post-impact symmetry variable. A natural question is whether this freedom can be used not only to generate periodic hybrid motions, but also to stabilize them. A representative periodic orbit produced by controlled exterior impacts is shown in Fig.~\ref{fig:periodic_orbit}. Our numerical experiments indicate, however, that in the absence of dissipation, reset actuation alone does not yield a convincing stabilizing regime.

This observation motivates the introduction of dissipation in the continuous-time dynamics between impacts. Accordingly, we consider the modified flow
\[
\dot q=\frac{\partial H}{\partial p}, \qquad
\dot p=-\frac{\partial H}{\partial q}-\alpha p,
\]
with $\alpha>0$, while retaining the same exterior impact geometry. The resulting picture is consistent with the geometric interpretation developed throughout the paper: the reset provides directional correction in the symmetry variable, while the dissipative flow induces contraction between impacts.

For the pendulum-on-a-cart example, we consider reset feedback laws at impact of the form
\[
p_x^+ =
\kappa_{p_\theta}(p_\theta^- - p_\theta^*)
+\kappa_{p_x}(p_x^- - p_x^*)
+\kappa_\theta(\theta^- - \theta^*),
\]
where $(\theta^*,p_\theta^*,p_x^*)$ denotes the corresponding point on the periodic orbit. Since impacts occur at fixed wall locations, the pre-impact position $x^-$ is fixed by the guard condition, and therefore no additional dependence on $x^-$ is included in the reset law. For Fig.~\ref{fig:eig_vals}, we restrict to the case $\kappa_{p_x}=0$, so that the reset law depends only on $p_\theta^-$ and $\theta^-$. Using $\alpha=1$, we identify regions in the $(\kappa_\theta,\kappa_{p_\theta})$-plane for which the dominant Floquet multiplier lies strictly inside the unit disk. These regions are shown in Fig.~\ref{fig:eig_vals}.

\begin{figure}[h!]
    \centering
    \includegraphics[width=\columnwidth]{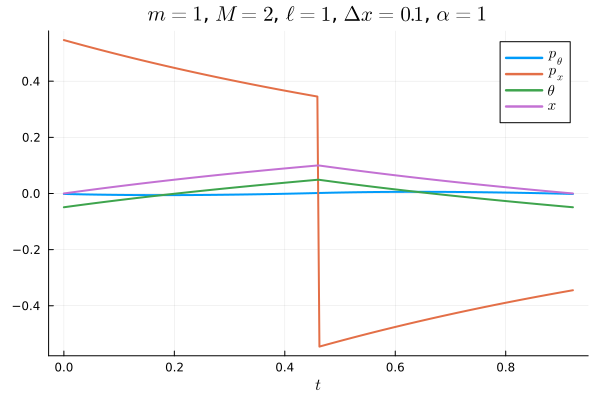}
    \caption{A periodic orbit generated by controlled exterior impacts. The distance between the two walls is $\varepsilon$.}
    \label{fig:periodic_orbit}
\end{figure}

\begin{figure}[h!]
    \centering
    \includegraphics[width=\columnwidth]{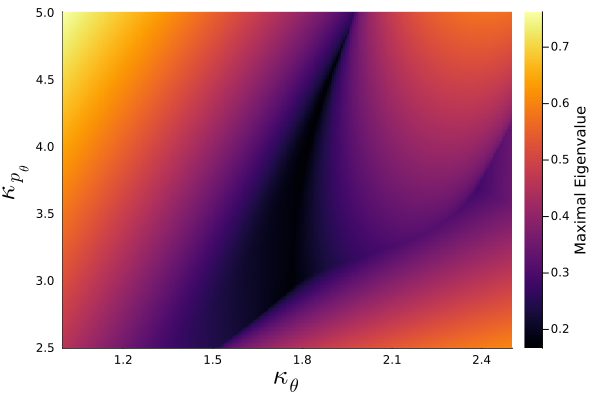}
    \caption{Largest Floquet multiplier for the reset law under consideration. In the computation shown here, $\kappa_{p_x}=0$, so the control input depends only on $p_\theta$ and $\theta$.}
    \label{fig:eig_vals}
\end{figure}

Using stabilizing gains identified from Fig.~\ref{fig:eig_vals}, we compute the basin of attraction shown in Fig.~\ref{fig:basin_attraction}. The basin is extremely thin, indicating that the stabilization mechanism is effective but geometrically delicate.

\begin{figure}[h!]
    \centering
    \includegraphics[width=\columnwidth]{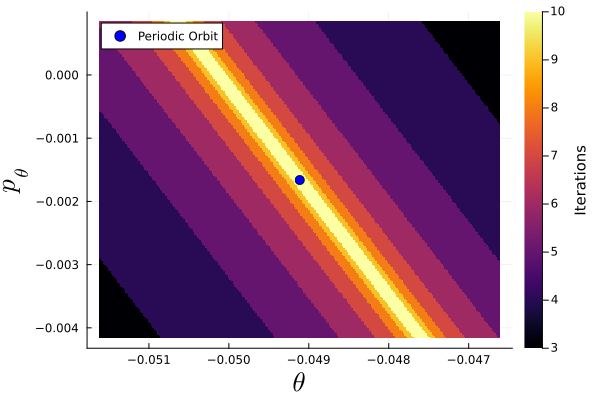}
    \caption{Heat map indicating the number of laps completed before destabilization.}
    \label{fig:basin_attraction}
\end{figure}

\subsection{A geometric criterion for impact-induced stabilization}

The stabilization mechanism identified above combines two ingredients: symmetry correction through the impact map and contraction along the continuous flow. The first is constrained by the geometry of the impact surface, while the second is provided here by dissipation. This leads to the following result.

\begin{theorem}\label{thm:geometric_criterion}
Let $\pi:Q\to Q/G$ be a principal $G$-bundle with mechanical connection $\mathcal A$, and consider a hybrid mechanical system with controlled reset maps $\Delta_{u,j}:TQ|_{\mathcal S_j}\to TQ$ along a periodic hybrid orbit $\gamma$, where $\mathcal S_j$ is the impact surface at the $j$-th collision. If every interior impact surface $\mathcal S_j$ along $\gamma$ is vertical, then each such reset preserves the mechanical connection, i.e. $\Delta_{u,j}^*\mathcal A=\iota_j^*\mathcal A$. Hence any reset-based stabilization mechanism relying on correction in the symmetry direction requires at least one non-vertical impact surface for which $\Delta_{u,j}^*\mathcal A\neq \iota_j^*\mathcal A$.

If, in addition, the continuous flow is replaced by a dissipative extension and the monodromy matrix $M_\gamma$ satisfies $\rho(M_\gamma)<1$, with $\rho$ the spectral radius, then $\gamma$ is exponentially stable.
\end{theorem}

\begin{proof}
For each interior impact surface $\mathcal S_j$, the hypothesis of verticality means that $\mathcal S_j=\pi^{-1}(\Sigma_j)$ for some codimension-one submanifold $\Sigma_j\subset Q/G$. By the interior-impact result established in Section~III.A, the corresponding reset preserves the mechanical connection, that is, $\Delta_{u,j}^*\mathcal A=\iota_j^*\mathcal A$, where $\iota_j:\mathcal S_j\hookrightarrow Q$ denotes the canonical inclusion. Hence, if $(q^-,\dot q^-)\in TQ|_{\mathcal S_j}$ and $(q^+,\dot q^+)=\Delta_{u,j}(q^-,\dot q^-)$, then $\mathcal A(q^+,\dot q^+)=\mathcal A(q^-,\dot q^-)$. Therefore the connection-induced symmetry variable is unchanged across every interior impact along $\gamma$.

It follows that, if all impacts encountered along $\gamma$ are interior, then the reset does not act on the symmetry direction through $\mathcal A$. Consequently, any reset-based stabilization mechanism that relies on correcting the symmetry component of the state must involve at least one non-vertical impact surface $\mathcal S_j$ for which $\Delta_{u,j}^*\mathcal A\neq \iota_j^*\mathcal A$.  

For the second claim, consider the closed-loop hybrid dynamics obtained after introducing dissipation in the continuous flow. Let $P_\gamma$ be the Poincar\'e return map associated with the periodic orbit $\gamma$, and let $M_\gamma = DP_\gamma(\gamma\cap \Sigma)$ denote its linearization on a local transverse section $\Sigma$. By definition, $M_\gamma$ is the monodromy matrix of the hybrid orbit. If $\rho(M_\gamma)<1$, then all Floquet multipliers lie strictly inside the unit disk. By the standard criterion for hybrid periodic orbits \cite{Westervelt2007},\cite{goodman2019existence} it follows that $\gamma$ is locally exponentially orbitally stable.
\end{proof}

In the pendulum-on-a-cart example, Theorem~\ref{thm:geometric_criterion} is reflected by the explicit reset formulas derived in Section~III. The interior impact surface $\mathcal S_{\mathrm{int}}=\{\alpha\}\times\mathbb R$ is vertical and the corresponding reset preserves the mechanical connection, $(\Delta_H^{\mathrm{int}})^*\mathcal A=\mathcal A|_{\mathcal S_{\mathrm{int}}}$. By contrast, for the exterior wall impacts one has $(\Delta_{H,v}^{\mathrm{ext}})^*\mathcal A\neq \mathcal A|_{\mathcal S_{\mathrm{ext}}}$ in general. Hence only the exterior impacts act nontrivially on the symmetry variable, while the dissipative flow provides the contraction needed for stabilization.

\section{Conclusions and Future Work}

We have shown that the geometry of the impact surface plays a decisive role in the use of impacts for control in hybrid mechanical systems with symmetry. In particular, interior impacts preserve the mechanical connection and therefore do not provide symmetry correction, whereas exterior impacts can act nontrivially on the symmetry variable. 

Future work will address more general dissipative mechanisms, such as Rayleigh-type dissipation, as well as hybrid mechanical systems with nonholonomic constraints. A central question is whether the correction-contraction mechanism identified here persists when symmetry, impacts, dissipation, and nonholonomic structure are simultaneously present.
\bibliographystyle{IEEEtran}
\bibliography{references}

\end{document}